\documentclass[preprint,11pt]{elsarticle}
\pdfoutput=1

\makeatletter
\def\ps@pprintTitle{%
 \let\@oddhead\@empty
 \let\@evenhead\@empty
 \def\@oddfoot{\centerline{\thepage}}%
 \let\@evenfoot\@oddfoot}
\makeatother

\usepackage{graphicx}
\usepackage{amsmath}
\usepackage{amssymb}
\usepackage{mathrsfs}
\usepackage{amsthm}
\usepackage{bm}
\usepackage{url}
\usepackage[T1]{fontenc}
\usepackage{csquotes}

\usepackage{setspace}
\MakeOuterQuote{"}
 \biboptions{sort&compress}

\newtheorem{theorem}{Theorem}
\newtheorem{lemma}{Lemma}

\theoremstyle{definition}

\theoremstyle{remark}
\newtheorem{remark}{Remark}

\newcommand{\mrm}[1]{\mathrm{#1}}
\newcommand{\tr}{\operatorname{tr}}

\newcommand{\diag}{\operatorname{diag}}

\newcommand{\rep}{\mathrel{\widehat{=}}}

\newcommand{\rmi}{\mathrm{i}}
\newcommand{\rme}{\mathrm{e}}

\newcommand{\rmT}{\mathrm{T}}

\newcommand{\be}{\begin{equation}}
\newcommand{\ee}{\end{equation}}
\newcommand{\ba}{\begin{align}}
\newcommand{\ea}{\end{align}}

\def\<{\langle}  
\def\>{\rangle}  

\newcommand{\bbF}{\mathbb{F}}
\newcommand{\bbQ}{\mathbb{Q}}
\newcommand{\bbZ}{\mathbb{Z}}
\newcommand{\Sp}[2]{\mrm{Sp}(#1,#2)}

\newcommand{\ASp}[2]{\mrm{ASp}(#1,#2)}
\newcommand{\SL}[2]{\mrm{SL}(#1,#2)}
\newcommand{\PSL}[2]{\mrm{PSL}(#1,#2)}

\newcommand{\soc}{\operatorname{soc}}

\newcommand{\Gal}{\mrm{Gal}}

\newcommand{\hw}{D}
\newcommand{\phw}{\overline{D}}

\newcommand{\pc}{\overline{\mathrm{C}}}

\def\eqref#1{\textup{(\ref{#1})}}  
\newcommand{\eref}[1]{Eq.~\textup{(\ref{#1})}}

\newcommand{\sref}[1]{Sec.~\ref{#1}}
\newcommand{\Sref}[1]{Section~\ref{#1}}

\newcommand{\thref}[1]{Theorem~\ref{#1}}
\newcommand{\Thref}[1]{Theorem~\ref{#1}}
\newcommand{\thsref}[1]{Theorems~\ref{#1}}

\newcommand{\lref}[1]{Lemma~\ref{#1}}
\newcommand{\Lref}[1]{Lemma~\ref{#1}}
\newcommand{\lsref}[1]{Lemmas~\ref{#1}}

\newcommand{\cref}[1]{Conjecture~\ref{#1}}
\newcommand{\Cref}[1]{Conjecture~\ref{#1}}

\newcommand{\rcite}[1]{Ref.~\cite{#1}}
\newcommand{\rscite}[1]{Refs.~\cite{#1}}

\begin{document}
\begin{frontmatter}
\title{Super-symmetric informationally complete measurements}
\author{Huangjun Zhu}
\address{Perimeter Institute for Theoretical Physics, Waterloo, On N2L 2Y5, Canada\\
\today}
\ead{hzhu@pitp.ca}


\begin{abstract}
Symmetric informationally complete measurements (SICs in short) are highly symmetric structures in the Hilbert space. They possess many nice properties which render them an ideal candidate for fiducial measurements. The symmetry of SICs is intimately connected with the geometry of the quantum state space and also has profound implications for foundational studies.   Here we explore those SICs that are most symmetric according to a natural criterion and show that all of them are covariant with respect to the Heisenberg-Weyl groups,  which are characterized  by the discrete analogy of the canonical commutation relation. Moreover,  their symmetry groups are subgroups of the Clifford groups. In particular, we prove that the SIC in dimension~2, the Hesse  SIC in dimension~3, and the set of Hoggar lines in dimension~8 are the only three SICs up to unitary equivalence whose symmetry groups act transitively on pairs of SIC projectors. Our work not only provides valuable insight  about SICs,   Heisenberg-Weyl groups,  and Clifford groups, but also offers a new approach and perspective for studying many other discrete symmetric structures behind finite state quantum mechanics, such as mutually unbiased bases and discrete Wigner functions.
\end{abstract}


\end{frontmatter}

%




\section{Introduction}
Symmetry plays a fundamental role in all areas of natural science. Of special interest are those objects possessing highest symmetry, which are the targets of constant quest. In this paper, we are concerned with an elusive discrete symmetric structure known as \emph{symmetric informationally complete measurements} (SICs in short) \cite{Zaun11,ReneBSC04, ScotG10,Zhu12the, ApplFZ15G}.
In a $d$-dimensional Hilbert space, a SIC is usually composed of $d^2$ subnormalized projectors onto pure states
$|\psi_j\rangle\langle\psi_j|/d$ with equal pairwise fidelity,
\begin{equation} \label{eq:SICinner}
|\langle\psi_j|\psi_k\rangle|^2=\frac{d\delta_{jk}+1}{d+1}.
\end{equation}
Here by a "SIC" we shall mean the set  of SIC
projectors $\Pi_j=|\psi_j\rangle\langle\psi_j|$, which sum up to $d$ times of the identity, $\sum_j\Pi_j=d$.
SICs have many nice properties rooted in their high symmetry, which make them an ideal candidate for fiducial measurements. They  play a crucial role in studying  quantum Bayesianism \cite{FuchS13} and in understanding the geometry of quantum state space \cite{BengZ06book,ApplEF11}.
They are useful in  linear quantum state tomography
\cite{Scot06,ZhuE11,Zhu12the, Zhu14IOC,Zhu14T,PetzR12E}, quantum cryptography
\cite{FuchS03,Rene04the,Rene05,EnglKNC04,DurtKLL08}, and signal processing~\cite{HowaCM06}.
They also have intriguing connections with many other interesting subjects, such as  equiangular lines, 2-designs,
mutually unbiased bases (MUB),  Lie algebras, and Galois theory; see \rcite{ApplFZ15G} and references therein.

A SIC gives rise to  a regular simplex in the operator space. Although this perspective is  fruitful in understanding its properties \cite{ApplFZ15G}, it has an obvious limitation: most permutations among the SIC projectors cannot be realized by unitary or antiunitary transformations, those transformations that preserve the quantum state space. This conflict between permutation symmetry and unitary symmetry  has profound implications for foundational studies  and  quantum information science \cite{BengZ06book, FuchS13, ApplEF11, ApplFZ14M,Zhu15P}.
It is closely connected to  the fact that the state space is not a ball except for dimension 2.
A better understanding of the symmetry of SICs is crucial to decoding the geometry of the quantum state space as well as foundational and practical issues entangled with the geometry.

 How symmetric is a SIC? This is a basic question we need to answer before we can conceive a clear picture about the quantum state space. Motivated by this question, here we determine those SICs that are most symmetric according to a natural criterion.  By virtue of the classification of finite simple groups (CFSG) \cite{ConwCNP85,Wils09book}, we show that the SIC in dimension~2, the Hesse  SIC in dimension~3 \cite{Zaun11,Appl05,Zhu10,Zhu12the, Hugh07, Beng10, TabiA13, DangBBA13}, and the set of Hoggar lines in dimension~8 \cite{Hogg98, Zaun11, Zhu12the} are the only three SICs up to unitary equivalence whose symmetry groups act transitively on pairs of SIC projectors. All these SICs are covariant with respect to  Heisenberg-Weyl (HW) groups,  and their symmetry groups are subgroups of  Clifford groups \cite{BoltRW61I, BoltRW61II, Appl05, Zhu10}. Moreover, only the Hesse SIC is covariant with respect to the  Clifford group. These results provide valuable insight on the elusive symmetry of SICs and geometry of the quantum state space.

Although our main focus here is the symmetry of SICs, it turns out that the approach introduced here  is also surprisingly useful  for
studying many other discrete symmetric structures behind finite state quantum mechanics, such as HW groups,   Clifford groups, MUB, discrete Wigner functions, and unitary 2-designs \cite{Zhu15P, Zhu15M}. For example, based on the present work, recently we have shown that the operator basis underlying the discrete Wigner function introduced by Wootters \cite{Woot87} is almost uniquely characterized by the double transitivity of its symmetry group \cite{Zhu15P}, which turns out to be equivalent to its symmetry group being a unitary 2-design \cite{Zhu15U}.  These conclusions establish the unique status of the Wootters discrete Wigner function over all other  quasi-probability representations of quantum mechanics. The ramifications of the present work are still under exploration.

The rest of the paper is organized as follows. In \sref{sec:stage} we introduce the background and the concept of super-symmetric informationally complete measurements (super-SICs for short). In \sref{sec:main} we state our main results. In \sref{sec:equi} we establish the equivalence of several symmetry properties of SICs. In \sref{sec:tripleP} we derive a necessary condition for the existence of super-SICs based on triple products among SIC projectors. In \sref{sec:superSICpp} we determine all super-SICs in prime power dimensions. In \sref{sec:superSIC-HW} we reveal a remarkable connection between super-SICs and the HW group, thereby determining all super-SICs. \Sref{sec:summary} summarizes this paper. Some technical details are relegated to the appendix.

\section{\label{sec:stage}Setting the stage}

\subsection{Symmetry of SICs}
The \emph{symmetry group}
$\overline{G}$  of a SIC $\{\Pi_j\}$ is composed of all unitary operators $U$
that leave the set of SIC projectors invariant; that is,
$U\Pi_jU^\dag=\Pi_{\sigma(j)}$ for a suitable permutation $\sigma$.  By convention, operators that differ only by overall phase factors are identified, as indicated by the overline notation in "$\overline{G}$".  The \emph{extended symmetry group} is the larger group that  contains also antiunitary operators.
Every SIC known so far is  \emph{group covariant} in the sense that
it can be generated from a single state---the \emph{fiducial state}---by a group composed of
unitary operators, that is, the SIC projectors form a
single orbit under the action of the symmetry group. In addition,
every  known  SIC is covariant with respect to one or another version of the HW group \cite{Zhu12the}, and this is true for every SIC in dimension 3 \cite{HughS14}. In particular, every known SIC is  sharply covariant in the sense that the generating group can be chosen to have the minimal possible order, that is, the number of SIC projectors. By contrast, sharply covariant MUB are quite rare; actually, only two examples are known \cite{Zhu15Sh, Zhu15N}.

How much additional symmetry can we expect beyond group covariance?  To understand the significance of  this question, it is instructive to compare a SIC in dimension $d$ with the regular simplex in the operator space defined by the SIC \cite{ApplFZ15G}. The isometry group of the regular simplex is isomorphic to the full permutation group of $d^2$ letters, which can realize any permutation among the vertices.  However, most of these permutations cannot be realized by unitary or antiunitary transformations, those transformations that  leave the  quantum state space invariant, because the  state space is not a ball except in dimension 2.  A natural question to ask is to what extent  the permutation symmetry can be retained given this limitation. Such surviving  symmetry is closely related to the geometry of the quantum state space, so any progress concerning this subject may potentially lead to a clearer picture about the quantum state space.

To answer the question posed above, we need to introduce several new concepts.
A SIC   is \emph{$k$-covariant} if  every ordered $k$-tuple of SIC projectors can be mapped to every other such $k$-tuple within its symmetry group; that is, its symmetry group acts $k$-transitively in the language of permutation groups \cite{DixoM96book, Came99book}; see \ref{sec:PermuGroup} for a short introduction. A $k$-covariant SIC is also referred to as doubly covariant when $k=2$ and triply covariant when $k=3$.  As we shall see shortly, no triply covariant SICs can exist. For this reason doubly covariant SICs are called \emph{super-symmetric} since they represent the most symmetric structure that can appear in the Hilbert space.

To complete the picture, a SIC  is \emph{$k$-homogeneous} if  every unordered $k$-tuple of SIC projectors can be mapped to every other such $k$-tuple under its symmetry group. A SIC in dimension $d$ is $k$-homogeneous if and only if it is $(d^2-k)$-homogenous. In addition, any $k$-covariant SIC is $k$-homogeneous (the converse is also true when $k\leq d^2/2$ according to \lref{lem:homoSuper} in \sref{sec:equi}, but is not so obvious). When the symmetry group is replaced by the extended symmetry group, the terminologies are modified by adding the prefix "quasi". For example, a SIC is quasi-super-symmetric if every two ordered pairs of SIC projectors  can be mapped to each other under  its extended symmetry group.

The concepts introduced above also apply to operator bases and generalized measurements \cite{Zhu15P}. In this paper, however, we shall focus on SICs.

\subsection{\label{sec:HWCli}Heisenberg-Weyl groups and Clifford groups}
To make this paper self-contained, here we present a short introduction about HW groups and Clifford groups; see \rscite{BoltRW61I, BoltRW61II, Gott97the, Appl05, Zhu10, Zhu12the} for more details.

The HW group is
 generated by the phase operator $Z$ and the cyclic shift operator
$X$ defined by their action on the kets $|e_r\rangle$ of the
"computational basis",
\begin{equation}\label{eq:HW}
Z|e_r\rangle=\omega^r|e_r\rangle, \qquad X|e_r\rangle=
|e_{r+1}\rangle,
\end{equation}
where $\omega=\mathrm{e}^{2\pi \mathrm{i}/d}$,  $r\in \bbZ_d$,
and $\bbZ_d$ is the ring of integers modulo $d$. The two
generators obey the  canonical commutation relation
\begin{equation}\label{eq:CCR1}
XZX^{-1}Z^{-1}=\omega^{-1},
\end{equation}
which determines the HW group up to unitary equivalence and overall
phase factors \cite{Weyl31book}.

In this paper, it turns out that another version of the HW group, called the multipartite HW group, is more relevant to our study. This HW group~$\hw$ is defined  only in prime power dimensions. It coincides with the HW group defined above in every prime dimension $p$.
In prime power dimension $q=p^n$, the multipartite  HW group  is the tensor power of $n$ copies of the HW group in  dimension $p$.
The elements in the multipartite HW are called displacement operators (or Weyl operators). Up to phase factors, they can be labeled by  vectors of length $2n$ defined over $\mathbb{F}_p=\bbZ_p$ as
\begin{equation}
D_{\mu}= \tau^{\sum_j \mu_j \mu_{n+j} }\prod_{j=1}^n X_j^{\mu_j}  Z_j^{\mu_{n+j}},
\end{equation}
where  $\tau=-\rme^{\pi \rmi/p}$,  while $Z_j$ and $X_j$ are the phase operator and cyclic shift operator of the $j$th party, as defined  in \eref{eq:HW} with $d=p$.  These operators satisfy the commutation relation
\begin{equation}
D_{\mu}D_{\nu} D_{\mu}^\dag D_{\nu}^\dag  =\omega^{\langle\mu,\nu\rangle},
\end{equation}
where $\langle\mathbf{\mu},\mathbf{\nu}\rangle=\mu^\rmT J\nu$ is the symplectic product with
$J=\bigl(\begin{smallmatrix}0_n &-1_n\\ 1_n& 0_n
\end{smallmatrix}\bigr)$.
Two displacement operators $D_{\mu}$ and $D_{\nu}$ commute if and only if the corresponding symplectic product $\langle\mathbf{\mu},\mathbf{\nu}\rangle$ vanishes.  The vectors $\mu$ together with the symplectic product form a symplectic space $\bbF_p^{2n}$ of dimension $2n$. The group of linear transformations that preserve the symplectic product is known as the symplectic group and denoted by $\Sp{2n}{p}$ \cite{Tayl92book}, which has order
\begin{equation}\label{eq:SpOrder}
p^{n^2}\prod_{j=1}^n(p^{2j}-1).
\end{equation}

The multipartite HW group defines a faithful irreducible projective representation of an elementary abelian group, recall that a group is elementary abelian if it is abelian and all elements other than the identity have the same order, which is a prime.
The converse is shown in the following lemma.
\begin{lemma}\label{lem:ProjEA}
Suppose $H$ is an elementary abelian group of order $p^{2n}$ with $p$ a prime and $n$ a positive integer. Then  every faithful irreducible projective representation of $H$ has degree $p^n$ and the image  is (projectively) unitarily equivalent to the multipartite HW group in dimension~$p^n$.
\end{lemma}
\begin{remark}
Although this lemma is known as folklore, it is difficult to find an explicit statement in the literature; see  \rscite{BoltRW61I, BoltRW61II, Morr73} for relevant conclusions. A self-contained proof is presented in the appendix.
\end{remark}

The Clifford group $\mathrm{C}$ is  the normalizer of the multipartite HW group, that is, the group composed of all unitary operators that leave the multipartite HW group invariant up  to phase factors \cite{BoltRW61I, BoltRW61II, Gott97the, Appl05, Zhu12the}. Any Clifford unitary $U$ induces a linear transformation on the symplectic space that labels the displacement operators. This linear transformation necessarily belongs to the symplectic group $\Sp{2n}{p}$ since conjugation preserves the commutation relation, so the induced linear transformation preserves the symplectic product.
Conversely, given any symplectic matrix $F$, there exist $q^2$ Clifford unitaries (up to phase factors) that induce  $F$, which differ from each other by displacement operators \cite{BoltRW61I,BoltRW61II}.
The quotient group $\pc/\phw$ ($\overline{G}$ denotes the group $G$ modulo phase factors) can be identified with the symplectic group $\Sp{2n}{p}$.
The Clifford group  $\pc$ in dimension $p^n$ has order
\begin{equation}\label{eq:CliffordOrder}
p^{n^2+2n}\prod_{j=1}^n(p^{2j}-1).
\end{equation}
Note that the symplectic group $\Sp{2}{p}$ in the special case $n=1$ can be identified with the special linear group $\SL{2}{p}$.   When $p$ is odd, $\pc$ is  isomorphic to the affine symplectic group $\ASp{2n}{p}=\Sp{2n}{p}\ltimes \bbF_p^{2n}$  \cite{BoltRW61I,BoltRW61II}.

\section{\label{sec:main}Main results}
After careful inspection of all  SICs known in the literature \cite{Zhu12the}, we find three  super-SICs (SICs that are super-symmetric). The first obvious candidate is the SIC in dimension~2, whose symmetry group can realize all even permutations among the SIC projectors (this SIC is also quasi-4-covariant, since the extended symmetry group can realize all permutations).

The second super-SIC is the \emph{Hesse SIC}  in dimension 3 \cite{Zaun11,Appl05,Zhu10,Zhu12the, Hugh07, Beng10, TabiA13, DangBBA13},  which is  generated by the HW group  from the fiducial ket
\begin{equation}\label{eq:GodSIC}
\frac{1}{\sqrt{2}}
(0, 1, -1)^\rmT.
\end{equation}
It is the only known SIC whose  symmetry group is the whole Clifford group, which has order $9\times 24$ in dimension 3 \cite{Zhu10}.
The Hesse SIC also has an intimate connection with the discrete Wigner function introduced by Wootters \cite{Woot87}. Let $\Pi_j$ be SIC projectors in the Hesse SIC, then  $1-2\Pi_j$ happen to be  phase point operators underlying the Wootters discrete Wigner function. This result is not a mere coincidence, but has deep reasons, as explained in \rcite{Zhu15P}.

The third  super-SIC is the set of \emph{Hoggar lines} \cite{Hogg98, Zaun11, Zhu12the} generated by the three-qubit Pauli group (the multipartite HW group in dimension 8) from
\begin{equation}\label{eq:HoggarLines}
\frac{1}{\sqrt{6}}(1+\rmi,0,-1,1,-\rmi,-1,0,0)^\rmT.
\end{equation}
It has an exceptionally large  symmetry group, which has order $64\times 6048$ \cite{Zhu12the}\footnote{The stabilizer of each fiducial state in the set of Hoggar lines is a nonabelian simple group of order 6048, which turns out to be isomorphic to the projective special unitary group $\mathrm{PSU}(3,3)$.}. It is the only SIC known so far that is not covariant with respect to the usual HW group defined in \eref{eq:HW} and also the only known SIC that is covariant with respect to the multipartite HW group in a dimension that is not prime \cite{Zhu12the}. In addition, since all SIC projectors are connected by local unitary transformations, they all have the same entanglement. Such SICs are interesting but quite rare in prime power dimensions, although there is another  example in dimension~4 \cite{ZhuTE10M, ZhuTE10T,Zhu12the}.

It turns out that the three super-SICs we have identified exhaust all  possibilities.
\begin{theorem}[CFSG]\label{thm:Super-SIC}
The SIC in dimension 2, the Hesse SIC in dimension~3, and the set of Hoggar lines in dimension~8 are the only three (quasi)-super-SICs up to unitary equivalence. They are also the only three (quasi)-2-homogeneous SICs up to unitary equivalence.
\end{theorem}

\begin{remark}
The full proof of \thref{thm:Super-SIC} relies on the classification of 2-transitive permutation groups \cite{Hupp57, Heri74,Heri85,DixoM96book, Came99book}, which in turn relies on the CFSG \cite{ConwCNP85,Wils09book}.
 To make room for future improvement, we mark all theorems and lemmas with "CFSG" whenever CFSG is involved in the proofs. In the case of prime power dimensions, nevertheless, we can prove \thref{thm:Super-SIC} without the CFSG; see \thref{thm:Super-SICpp} in \sref{sec:superSICpp}. In addition, we  try to prove as much as possible without resorting to such heavy tools; the CFSG is required directly only in the proof of \lref{lem:2transitiveAS} in \sref{sec:superSIC-HW}.
\end{remark}

In sharp contrast with SICs, all the operator bases underlying the Wootters discrete Wigner functions \cite{Woot87} are super-symmetric \cite{Zhu15P}. \Thref{thm:Super-SIC} shows that the restriction to pure states imposes a stringent limitation on the potential symmetry of an operator basis. This observation is instructive to understanding quasi-probability representations of quantum mechanics and may have implications for foundations studies, such as quantum Bayesianism \cite{FuchS13, ApplFZ14M}.

\section{\label{sec:equi}Equivalence of 2-homogeneous SICs and super-SICs}

Although the requirement of $k$-covariance is a priori much stronger than that of $k$-homogeneity whenever $k>1$, it turns out that the two requirements are equivalent for SICs as long as $k\leq d^2/2$.
\begin{lemma}\label{lem:homoSuper}
A SIC in dimension $d$ is (quasi)-$k$-covariant with $k\leq d^2/2$ if and only if it is (quasi)-$k$-homogeneous. In particular,
a SIC is (quasi)-super-symmetric if and only if it is (quasi)-2-homogeneous.
\end{lemma}
\begin{proof}The claims are obvious when $k=1$. When $k\geq2$, according to Theorem 1 of Kantor~\cite{Kant72} and Theorem 9.4B in \rcite{DixoM96book},  any $k$-homogenous permutation group  of degree $n\geq2k$ is ($k-1$)-transitive. If $n$ is a perfect square, then the group
 is also $k$-transitive, except possibly some  4-homogenous permutation group $G$ of degree~9.  In addition, such a group $G$ must contain a subgroup isomorphic to $\PSL{2}{8}$, which is a nonabelian simple group. These facts confirm the lemma immediately except when the dimension is equal to 3 and the SIC is (quasi)-4-homogeneous. In this special case, the SIC is (quasi)-triply covariant and is thus group covariant.
 Any group covariant   SIC in dimension 3 is covariant with respect to the HW group, and its extended symmetry group is a subgroup of the extended Clifford group~\cite{Zhu10}; actually, this is  true for any SIC in dimension 3 according to \rcite{HughS14}.  However, the extended Clifford group  in dimension 3 is solvable and thus cannot contain any subgroup isomorphic to $\PSL{2}{8}$. Alternatively, this special case can be excluded by \lref{lem:3covariantQuasi} in \sref{sec:tripleP}.
\end{proof}
In view of \lref{lem:homoSuper},
the second part of  \thref{thm:Super-SIC} is an immediate consequence of the first part. Therefore, we can focus on (quasi)-super-SICs in the rest of the paper.

When the dimension is even, we have another equivalence relation.
\begin{lemma}\label{lem:quasi-super-SIC}
Every quasi-super-SIC in an even dimension is a super-SIC.
\end{lemma}
\begin{remark}
Although this lemma also holds when the dimension is odd, so far we can prove this conclusion only after the classification of all quasi-super-SICs.
\end{remark}
\Lref{lem:quasi-super-SIC} is an immediate consequence of the following lemma, note that the symmetry group of a SIC is a subgroup of the extended symmetry group of index at most 2.
\begin{lemma}
Any index-2 subgroup of a 2-transitive permutation group of even degree at least 4 is 2-transitive.
\end{lemma}
\begin{proof}
Suppose $H$ is an index-2 subgroup of the 2-transitive permutation group $G$ on $\Omega$ with $|\Omega|\geq 4$.
Then $H$ is normal in $G$ and acts transitively on $\Omega$.
Let $S\in G$ be the stabilizer of a given point  in $\Omega$, say $\alpha$,  then $S$ acts transitively on $\Omega\setminus \alpha$. Let $R=S\cap H$; then $R$ is an index-2 normal subgroup of $S$. Therefore, $R$ is either transitive  on $\Omega\setminus \alpha$ or has two orbits of equal length. When the degree $|\Omega|$  is even, however, the latter scenario cannot happen since the cardinality of $\Omega\setminus \alpha$ is odd. So $R$ is transitive on $\Omega\setminus \alpha$, and $H$ is 2-transitive on $\Omega$.
\end{proof}

\section{\label{sec:tripleP}Symmetry and triple products}

Before proving our main result, we first show that no triply covariant SICs can exist and that for prime-power dimensions, (quasi)-super-SICs can only exist in dimensions 2, 3, and 8. The proofs are  based on simple observation on triple products among SIC projectors,
\begin{equation}
T_{jkl}=\tr(\Pi_j\Pi_k\Pi_l).
\end{equation}
Such triple products  have played an important role in studying the symmetry of and equivalence relations among SICs \cite{ApplFF11,Zhu10,Zhu12the}. They are also useful to studying discrete Wigner functions \cite{Woot87}.
Note that $|T_{jkl}|$ is equal to $(d+1)^{-3/2}$ if the three indices are distinct, while it is equal to $1/(d+1)$ or 1 if two or three indices coincide.
The normalized triple products  $\tilde{T}_{jkl}=T_{jkl}/|T_{jkl}|$ satisfy
\begin{gather}
\tilde{T}_{jkl}=\tilde{T}_{klj}=\tilde{T}_{ljk}=\tilde{T}_{jlk}^*=\tilde{T}_{lkj}^*=\tilde{T}_{kjl}^*, \label{eq:TripleProdP1}\\
\tilde{T}_{jkl}=\tilde{T}_{mjk}\tilde{T}_{mkl}\tilde{T}_{mlj}.\label{eq:TripleProdP2}
\end{gather}

\begin{lemma}\label{lem:3covariant}
No SIC is triply covariant.
\end{lemma}
\begin{remark}
If a SIC is triply covariant, then all triple products among distinct SIC projectors are equal. Consequently, all permutations among SIC projectors can be realized by unitary transformations according to Theorem~3 in \rcite{ApplFF11} (see also Chap.~10 in \rcite{Zhu12the}). Such a SIC would be too symmetric to exist!
\end{remark}
\begin{proof}
 Suppose on the contrary that
the SIC $\{\Pi_j\}$  is triply covariant. Then each $\tilde{T}_{jkl}$ for distinct $j,k,l$  equals  $\pm 1$, where the sign is independent of $j,k,l$, so
\begin{equation}
\sum _{l} T_{jkl}=\pm \frac{d^2-2}{(d+1)^{3/2}}+\frac{2}{d+1}, \quad  j\neq k.
\end{equation}
Since the SIC projectors $\Pi_l$ sum up to $d$ times of the identity, we also have
\begin{equation}
\sum _{l} T_{jkl}=d\tr(\Pi_j\Pi_k)=\frac{d}{d+1}.
\end{equation}
However, the  two equations above  cannot be satisfied simultaneously. This contradiction completes the proof.
\end{proof}

\begin{lemma}\label{lem:3covariantQuasi}
No SIC in dimension   $d\geq 3$ is quasi-triply covariant.
\end{lemma}
Since this lemma is not essential in proving our main result, the proof is relegated to the appendix.

\begin{lemma}\label{lem:SuperSICpp}
Suppose there exists a quasi-super-SIC in prime-power dimension~$d$. Then $d$ equals 2, 3, or 8.
\end{lemma}
This lemma follows from  \lsref{lem:TripleProdRoot}, \ref{lem:SuperSIC}, and \ref{lem:quadraticEx} below.

\begin{lemma}\label{lem:TripleProdRoot}
Suppose $\{\Pi_j\}$ is a quasi-super-SIC with normalized triple products $\tilde{T}_{jkl}$. Then all $\tilde{T}_{jkl}$ are $2d^2$th roots of unity, that is, $\tilde{T}_{jkl}^{2d^2}=1$.
\end{lemma}
\begin{proof}
If  $\{\Pi_j\}$ is  super-symmetric, then the multiset $\{\tilde{T}_{mjk}\}_{m=1}^{d^2}$ is identical with $\{\tilde{T}_{mkj}\}_{m=1}^{d^2}$. On the other hand, the two multisets are conjugates of each other. Therefore, both of them are conjugation invariant and $\prod_{m} \tilde{T}_{mjk}= \pm 1$, where the sign is independent of $j$ and $k$ as long as they are distinct. Now taking the product over $m$ in \eref{eq:TripleProdP2} yields $\tilde{T}_{jkl}^{d^2}= \pm 1$, which implies the lemma.

If  $\{\Pi_j\}$ is  quasi-super-symmetric but not super-symmetric, then $\{\Pi_j\}$ is  quasi-2-homogenous  but not 2-homogenous according to \lref{lem:homoSuper}. It follows that every ordered pair of distinct SIC projectors can be mapped to the same pair with the reverse order under the  symmetry group. So the multiset $\{\tilde{T}_{mjk}\}_{m=1}^{d^2}$ is still conjugation invariant and the lemma holds as before.
\end{proof}

\begin{lemma}\label{lem:SuperSIC}
Suppose there exists a quasi-super-SIC in dimension $d\geq3$, then $(d-2)\sqrt{d+1}\in \bbZ[\zeta_{2d^2}]$ and $\sqrt{d+1}\in \bbQ[\zeta_{2d^2}]$, where  $\zeta_{2d^2}$ is a primitive $2d^2$th root of unity.
\end{lemma}
\begin{remark}
Here $\bbZ[\zeta_{2d^2}]$ is the extension of the  ring of  integers by $\zeta_{2d^2}$, and $\bbQ[\zeta_{2d^2}]$ is the extension of the field of rational numbers  by $\zeta_{2d^2}$ \cite{Ash07book}.
\end{remark}
\begin{proof}
From the two equations
\begin{equation}
\begin{aligned}
\sum_{k} \tr(\Pi_j\Pi_k\Pi_l)&=\frac{2}{d+1}+\sum_{k\neq j,l}\frac{1}{(d+1)^{3/2}}\tilde{T}_{jkl},  \\
\sum_{k} \tr(\Pi_j\Pi_k\Pi_l)&=d\tr(\Pi_j\Pi_l)=\frac{d}{d+1}, \quad j\neq l,
\end{aligned}
\end{equation}
we deduce that
\begin{equation}
\sum_{k\neq j,l}\tilde{T}_{jkl}=(d-2)\sqrt{d+1}.
\end{equation}
The lemma follows from the fact that $\tilde{T}_{jkl}$ are $2d^2$th roots of unity according to \lref{lem:TripleProdRoot}.
\end{proof}

Now \lref{lem:SuperSICpp} is a consequence of \lref{lem:SuperSIC} and the following lemma.
\begin{lemma}\label{lem:quadraticEx}
Suppose $d$ is a prime power. Then $\sqrt{d+1}\in \bbQ[\zeta_{2d^2}]$ if and only if $d= 3, 8$.
\end{lemma}
\begin{proof}
Suppose $\sqrt{d+1}\in \bbQ[\zeta_{2d^2}]$. Then $\sqrt{d+1}$ belongs to $\bbQ$ or a quadratic extension of $\bbQ$. In the former case $d+1$ must be a perfect square, which implies that $d=3,8$ given that $d$ is a prime power (cf. \lref{lem:perfectS1} in the appendix).

Observing that  $\bbQ[\zeta_{2d^2}]$ is a Galois extension of $\bbQ$, we conclude that the  quadratic extensions  of $\bbQ$ contained in
$\bbQ[\zeta_{2d^2}]$ are in one-to-one correspondence with the subgroups of the Galois group $\Gal(\bbQ[\zeta_{2d^2}]/\bbQ)$ of index 2~\cite{Ash07book}. The group $\Gal(\bbQ[\zeta_{2d^2}]/\bbQ)$ is isomorphic to the automorphism group of $\bbZ_{2d^2}$, which in turn is isomorphic to the multiplicative group $\bbZ_{2d^2}^*$ of invertible elements (or units) in $\bbZ_{2d^2}$~\cite{KurzS04book}.

When $d$ is a power of 2,  $\bbZ_{2d^2}^*$ is isomorphic to $\bbZ_2\times \bbZ_{d^2/2}$, which contains three  subgroups of index 2. Consequently,  $\bbQ[\zeta_{2d^2}]$ contains three quadratic extensions over $\bbQ$, which are identical with the three quadratic extensions over $\bbQ$ contained in $\bbQ[\zeta_{8}]$. It is easy to verify that the three quadratic extensions are  $\bbQ[\sqrt{2}]$, $\bbQ[\rmi]$, and $\bbQ[\sqrt{2}\rmi]$, none of which can contain $\sqrt{d+1}$ except when $d=8$.

When $d=p^n$ is a power of an odd prime $p$, $\bbZ_{2d^2}^*$  is cyclic of order $p^{2n-1}(p-1)$ and contains a unique subgroup of index 2. Consequently,  $\bbQ[\zeta_{2d^2}]$ contains a unique  quadratic extension over $\bbQ$, which is identical with the  quadratic extension over $\bbQ$ contained in $\bbQ[\zeta_{p}]$. According to the famous quadratic reciprocity (see Proposition 7-3-1 in \rcite{Weis63book} for example), this unique quadratic extension is $\bbQ[\sqrt{(-1)^{(p-1)/2}p}]$, which cannot contain $\sqrt{d+1}$ except when $d=p=3$.
\end{proof}

\section{\label{sec:superSICpp}Super-SICs in prime power dimensions}
In this section we prove \thref{thm:Super-SIC} in the case of prime power dimensions without resorting to the CFSG.
\begin{theorem}\label{thm:Super-SICpp}
In prime power dimensions, the SIC in dimension 2, the Hesse SIC in dimension 3, and the set of Hoggar lines in dimension 8 are the only three (quasi)-super-SICs  up to unitary equivalence. They are also the only three (quasi)-2-homogeneous SICs up to unitary equivalence.
\end{theorem}

To achieve our goal, we need to introduce several basic  results concerning permutation groups \cite{DixoM96book,Came99book}; see \ref{sec:PermuGroup} for a short introduction on this subject.
\begin{theorem}[Burnside]\label{thm:Burnside}
Any  2-transitive permutation group $G$ on a finite set $\Omega$ has a unique minimal normal subgroup,
 which is either  elementary abelian acting regularly on or
 nonabelian simple  acting primitively on $\Omega$.
\end{theorem}
\begin{remark}
A \emph{minimal normal subgroup} $N$ of a  group $G$ is a normal subgroup other than the identity that contains no
other nontrivial normal subgroup of $G$.
A group action is \emph{regular} if it is transitive with trivial point stabilizer;
it is \emph{primitive} if it is transitive and preserves no nontrivial partition.
If $N$ is regular elementary abelian, then it can be identified as a vector space over some  finite Galois field, and $G$  as a subgroup of the group of affine semilinear transformations on the  vector space and is called   \emph{affine}. If $N$ is nonabelian simple, then $G$ can be identified as a subgroup of the automorphism group of  $N$ and is called \emph{almost simple} \cite{DixoM96book,Came99book}.
\end{remark}

Another useful tool in our study is the following lemma reproduced from Lemma 7.2 and Theorem 7.3 in the author's thesis \cite{Zhu12the} (see also  \rcite{Zhu15Sh}  and Theorem 2.34 in \rcite{Zaun11}).
\begin{lemma}\label{lem:orbits}
Suppose $\overline{G}$ is a subgroup of the symmetry group of a SIC. Then
the number of orbits  of
$\overline{G}$ on the SIC is equal to the sum of squared multiplicities of all the
inequivalent irreducible components of $\overline{G}$. In particular, $\overline{G}$ acts transitively on the SIC if and only if it is irreducible.
\end{lemma}

\begin{proof}[Proof of \thref{thm:Super-SICpp}]
In view of  \lsref{lem:homoSuper} and \ref{lem:SuperSICpp}, to  prove \thref{thm:Super-SICpp}, it remains to show the uniqueness of  quasi-super-SICs in dimensions 2,  3 and 8. Note that every quasi-super-SIC is necessarily group covariant.

In dimension 2, every SIC defines a regular tetrahedron on the Bloch sphere, so all SICs are unitarily equivalent.

In dimension~3,  every (group covariant) SIC  is covariant with respect to the HW group and its symmetry group is a subgroup of the  Clifford group, which is isomorphic  to $\mathrm{SL}(2, 3)\ltimes (\bbZ_3)^2$~\cite{Zhu10, Zhu12the, HughS14}. Therefore, the stabilizer of each SIC projector is isomorphic to a subgroup of $\mathrm{SL}(2, 3)$. If the SIC is quasi-super-symmetric, then the order of the stabilizer (within the symmetry group) is divisible by 4. Observing that $\mathrm{SL}(2, 3)$ contains a unique element of order~2, we conclude that the stabilizer contains an element of order 4. Since all order-4  Clifford unitary transformations in dimension 3 are conjugate to each other \cite{Zhu10}, without loss of generality we may assume that one fiducial ket is stabilized by (is an eigenket of) the order-4 Clifford unitary transformation
\begin{equation}
\frac{1}{\sqrt{3}}\begin{pmatrix}
1 & 1& 1\\
1 & \omega & \omega^2\\
1 &\ \omega^2 & \omega
\end{pmatrix},
\end{equation}
which happens to be the Fourier matrix in dimension 3.
Calculation shows that this  unitary transformation has three nondegenerate eigenkets, one of which happens to be the fiducial ket that generates the Hesse SIC (see \eref{eq:GodSIC}), while the other two are not fiducial kets. Therefore,  the Hesse  SIC is  the only quasi-super-SIC in dimension 3 up to unitary equivalence.

It remains to consider quasi-super-SICs in dimension 8. Suppose $\{\Pi_j\}$ is a quasi-super-SIC in dimension 8. According to \lref{lem:quasi-super-SIC}, $\{\Pi_j\}$  is super-symmetric, so its
symmetry group $\overline{G}$  is a 2-transitive (and automatically primitive) permutation group of degree 64.
  According to \thref{thm:Burnside}, $\overline{G}$ is either affine or almost simple and has a unique  minimal normal subgroup, say $\overline{N}$, which is either regular elementary abelian or nonabelian simple.
 According to Appendix B in \rcite{DixoM96book}, all 2-transitive permutation groups of degree 64 are of affine type except for the symmetric group and alternating group, which are 64 and 62 fold transitive, respectively.
In view of  \lref{lem:3covariant}, $\overline{G}$ can be isomorphic to neither the symmetric group nor the alternating  group of degree~64. Therefore, $\overline{G}$ is an affine 2-transitive permutation group, and $\overline{N}$ is regular elementary abelian. The group $\overline{N}$ is necessarily irreducible according to \lref{lem:orbits} and thus defines a faithful irreducible projective representation of  an elementary abelian group. According to \lref{lem:ProjEA}, $\overline{N}$ is unitarily equivalent to the multipartite HW group in dimension 8, that is, the three-qubit Pauli group. Consequently,  $\overline{G}$ is a subgroup of the Clifford group, which has order $2^{15}\cdot 3^4\cdot 5\cdot 7$ according to \eref{eq:CliffordOrder}.

Suppose $|\psi\rangle$ is a fiducial ket of the three-qubit Pauli group  that generates a quasi-super-SIC. Then its stabilizer  contains an order-7 Clifford unitary transformation. Observe that all Sylow 7-subgroups of the Clifford group are cyclic of order 7 and are conjugate to each other. Without loss of generality, we may assume that $|\psi\rangle$ is stabilized by the order-7 Clifford unitary transformation \cite{Zhu12the}
\begin{equation}\label{eq:HoggarStab}
U_7\rep\frac{\omega^5}{\sqrt{2}}\begin{pmatrix}
 0 & 0 & 1 & 0 & -\rmi & 0 & 0 & 0 \\
 0 & 0 & \rmi & 0 & -1 & 0 & 0 & 0 \\
 0 & 0 & 0 & -\rmi & 0 & -1 & 0 & 0 \\
 0 & 0 & 0 & -1 & 0 & -\rmi & 0 & 0 \\
 1 & 0 & 0 & 0 & 0 & 0 & -\rmi & 0 \\
 -\rmi & 0 & 0 & 0 & 0 & 0 & 1 & 0 \\
 0 & -\rmi & 0 & 0 & 0 & 0 & 0 & -1 \\
 0 & 1 & 0 & 0 & 0 & 0 & 0 & \rmi
\end{pmatrix},
\end{equation}
where $\omega=\rme^{2\pi \rmi/8}$. Calculation shows that $U_7$ has six nondegenerate eigenkets, none of which are fiducial kets.
The two-dimensional eigenspace corresponding to the eigenvalue 1 contains two fiducial kets, which
happen to be  the only two normalized kets in the eigenspace that satisfy the following three equations,
\begin{align}
\langle\psi| \sigma_z \otimes 1\otimes 1|\psi\rangle=\pm\frac{1}{3}, \; \langle\psi| 1\otimes \sigma_z \otimes  1|\psi\rangle=\pm\frac{1}{3}, \;
 \langle\psi| 1\otimes \sigma_x \otimes 1|\psi\rangle=\pm\frac{1}{3}.
\end{align}
The first fiducial ket happens to be the one that generates the set of Hoggar lines (see \eref{eq:HoggarLines}). The second one
\begin{equation}
\frac{1}{\sqrt{6}}(-\rmi, -1, 0, 0, -1 + \rmi, 0, 1, 1)^{\rmT}
\end{equation}
 generates a SIC which turns out to be equivalent to the set of Hoggar lines under the Clifford unitary transformation
\begin{equation}
\begin{pmatrix}0& U\\ V&0 \end{pmatrix},\quad  U=\diag(-\rmi,-\rmi,-1,1),\quad V=\diag(1,1,\rmi,-\rmi).
\end{equation}
This transformation was computed using the algorithm described in Chap.~10 of the author's thesis \cite{Zhu10}. Therefore, the set of Hoggar lines is the unique quasi-super-SIC in dimension 8 up to unitary equivalence.
\end{proof}

\section{\label{sec:superSIC-HW}Super-SICs and the Heisenberg-Weyl group}
In this section we prove our main result  \thref{thm:Super-SIC}. In view of \thref{thm:Super-SICpp}, it suffices to show that super-SICs can only exist in prime power dimensions. Here we not only prove this conclusion but also
 establish a remarkable  connection between super-SICs and the multipartite HW group, which is of independent interest.
\begin{theorem}[CFSG]\label{thm:SuperSIC-HW}
Every quasi-super-SIC is covariant with respect to the multipartite HW group in a prime power dimension; its (extended) symmetry group is a subgroup of the (extended) Clifford group.
\end{theorem}
\begin{remark}
Surprisingly, the symmetry requirement on a SIC naturally leads to the canonical commutation relation and the multipartite HW group.
When the dimension is a prime, this conclusion is consistent with the earlier conclusion of  the author \cite{Zhu10} that every group covariant SIC is covariant with respect to the HW group and that its (extended) symmetry group is a subgroup of the (extended) Clifford group.
 Recently, we have generalized \thref{thm:SuperSIC-HW} to operator bases, which plays a crucial role in understanding the discrete Wigner function \cite{Zhu15P}. In addition, \thref{thm:SuperSIC-HW} is useful to studying unitary 2-designs \cite{Zhu15U}.
\end{remark}

To prove \thref{thm:SuperSIC-HW}, we need two additional technical lemmas, whose  proofs are relegated to the appendix.
 \begin{lemma}\label{lem:index2Normalsub}
Suppose $H$ is a subgroup of index 2 of a 2-transitive permutation group $G$ with $|G|>2$ on $\Omega$. Then $H$ has a unique minimal normal subgroup, which coincides with the unique minimal normal subgroup of $G$.
\end{lemma}

\begin{lemma}[CFSG]\label{lem:2transitiveAS}
Suppose $G$ is an almost simple 2-transitive permutation group whose degree $n$ is a perfect square. Let $N$ be the minimal normal subgroup of $G$. Then one of the following three cases holds.
\begin{enumerate}
\item $N$ is isomorphic to the alternating group $A_n$ with $n\geq5$ a perfect square.
\item  $N$ is isomorphic to $\mathrm{PSL}(k,q)$ with $(k,q)=(2,8)$, $(4,7)$, or $(5,3)$, and $n$ is equal to $3^2$, $20^2$, or $11^2$ accordingly.
 \item  $N$ is isomorphic to $\mathrm{Sp}(6,2)$, and $n$ is equal to $6^2$.
\end{enumerate}
\end{lemma}
\begin{remark}Here the symbols $n,q$ are independent of  those appearing in \sref{sec:HWCli}.
The proof of this lemma relies on the classification of almost simple 2-transitive permutation groups \cite{DixoM96book,Came99book}, which in turn relies on the CFSG \cite{ConwCNP85,Wils09book}. All other results in this paper that rely on the CFSG rely on this lemma either directly or indirectly.
\end{remark}

\begin{lemma}[CFSG] \label{lem:MNS2}
The symmetry group and extended symmetry group of any quasi-super-SIC have a unique and identical   minimal normal subgroup, which is elementary abelian and regular.
\end{lemma}
\begin{proof}
Let $\{\Pi_j\}$  be a quasi-super-SIC with symmetry group $\overline{G}$ and extended symmetry group $\overline{EG}$. Then  $\overline{G}$ is either identical with $\overline{EG}$ or is a subgroup of index 2. So  $\overline{G}$ and  $\overline{EG}$  have a unique and common minimal normal subgroup according to \thref{thm:Burnside} and \lref{lem:index2Normalsub}. In addition, the minimal normal subgroup  $\overline{N}$ acts transitively on the set of SIC projectors and is thus irreducible according to \lref{lem:orbits}.

Suppose on the contrary that $\overline{N}$ is not elementary abelian or regular. Then
$\overline{N}$ is a nonabelian simple group, and  $\overline{EG}$ is an almost simple 2-transitive permutation group of degree $d^2$, which is a perfect square.  According to \lref{lem:3covariant}, $\overline{N}$ cannot be isomorphic to  $A_n$ for $n\geq 5$, because $A_n$ is $(n-2)$-transitive \cite{DixoM96book, Came99book}.
In view of  \lref{lem:2transitiveAS},
$\overline{N}$ is a faithful irreducible projective representation (over the complex numbers) of $\mathrm{PSL}(k,q)$ with $(k,q)=(2,8)$, $(4,7)$,  $(5,3)$, or $\mathrm{Sp}(6,2)$, and the degree of the representation (which equals the square root of the degree of the permutation group)  is 3, 20, 11, or 6, respectively.  According to Table II in \rcite{TiepZ96}, however, the minimal degree of such a representation is   7, 399, 120, or 7.  The same reasoning can also be used to exclude the alternating group $A_n$ (without relying on \lref{lem:3covariant}) given  that the minimal degree of such a representation for $A_n$ is $n-1$ when $n\geq 8$ \cite{BessNOT15}.
This contradiction completes the proof.
\end{proof}

\begin{proof}[Proof of \thref{thm:SuperSIC-HW}]
Suppose $\{\Pi_j\}$ is a quasi-super-SIC in dimension $d$ with symmetry group $\overline{G}$ and extended symmetry group $\overline{EG}$. Then $\overline{G}$ and $\overline{EG}$ have a unique common minimal normal subgroup, say $\overline{N}$. In addition,
$\overline{N}$ is  elementary abelian and acts regularly on the set of SIC projectors according to \lref{lem:MNS2}. Therefore, $\overline{N}$ has order $d^2$ and is irreducible  according to \lref{lem:orbits}. In a word, $\overline{N}$ is a faithful irreducible projective representation of an elementary abelian group, so it is (projectively) unitarily  equivalent to  the multipartite  HW group according to \lref{lem:ProjEA}. Given that $\overline{N}$ is normal in $\overline{G}$ and $\overline{EG}$, we conclude that  $\overline{G}$ ($\overline{EG}$) is a subgroup of the  Clifford group (extended Clifford group).
\end{proof}

Finally, we can determine all super-SICs, thereby achieving our main goal.
\begin{proof}[Proof of  \thref{thm:Super-SIC}]
\Thref{thm:Super-SIC} is an immediate consequence of \thsref{thm:Super-SICpp} and \ref{thm:SuperSIC-HW}.
\end{proof}

\section{\label{sec:summary}Summary}
We have introduced super-SICs as the most symmetric structure that can appear in the Hilbert space. We proved that the SIC in dimension~2, the Hesse SIC in dimension~3, and the set of Hoggar lines in dimension~8 are the only three super-SICs up to unitary equivalence.  Such general statements are of intrinsic interest but are quite rare in the literature due to  the enormous difficulty in decoding elusive SICs.
Our work provides valuable insight on   symmetry of SICs and  geometry of the quantum state space, which  may have implications for  foundational studies, such as  quantum Bayesianism.
Our work also reveals an intriguing connection between symmetry and the canonical commutation relation, which deserves further study. In addition, the ideas and techniques introduced here  are quite helpful  to
studying  MUB, discrete Wigner functions, and unitary 2-designs etc.
Furthermore, our work establishes diverse links between SICs  and various other subjects, such as number theory, representation theory, combinatorics, and theory of permutation groups, which are of interest to researchers from respective fields.

\section*{Acknowledgements}
The author is grateful to Dragomir \v{Z} {\DJ}okovi\'{c} for simplifying the proof of \lref{lem:index2Normalsub}, to Gergely Harcos for helping proving \lref{lem:NageL} in the appendix, and to Daniel El-Baz for recommending \rcite{Ribe94}.
The author also  thanks Marcus Appleby, Ingemar Bengtsson, Lin Chen,  Markus Grassl, and Mark Howard for discussions and suggestions. This work is supported in part by Perimeter Institute for Theoretical Physics. Research at Perimeter Institute is supported by the Government of Canada through Industry Canada and by the Province of Ontario through the Ministry of Research and Innovation.

\appendix

\section{Proof of \lref{lem:ProjEA}}
\begin{proof}
This lemma is  closely related to  Weyl's theorem that the HW group is uniquely characterized by the discrete analogy of the canonical commutation relation \cite{Weyl31book}.
Suppose $H$ has a $d$-dimensional faithful irreducible projective representation $A\mapsto U_A$ for $A\in H$. Let $A_1$ be an arbitrary element in $H$ other than the identity. Given that the representation is faithful, $U_{A_1}$ cannot be proportional to the identity and thus cannot commute with all  elements in the representation according to  Schur's lemma. So there exists an element $B_1\in H$ such that
\begin{equation}
U_{A_1}U_{B_1} U_{A_1}^\dag U_{B_1}^\dag =\rme^{\rmi\phi}
\end{equation}
with $\rme^{\rmi\phi}\neq 1$  a phase factor. Observing that $U_{B_1}^p$ is proportional to the identity, we conclude that $\rme^{\rmi\phi}$ is a $p$th root of unity. Replacing $B_1$ with a suitable power if necessary, we may assume that
\begin{equation}
U_{A_1}U_{B_1} U_{A_1}^\dag U_{B_1}^\dag =\omega^{-1},
\end{equation}
where $\omega=\rme^{2\pi\rmi/p}$ is a primitive $p$th root of unity. Note that $A_1$ and $B_1$ generate a group $H_1$ of order $p^2$. If $n=1$, then $H=H_1$,  the representation must have degree $p$, and the image is (projectively) unitarily  equivalent to the HW group in dimension $p$ according to Weyl's theorem \cite{Zhu10}.

If $n>1$,  let $A_2$ be an arbitrary element in $H$ not contained in $H_1$. Then $U_{A_2}$ commutes with $U_{A_1}$ and $U_{B_1}$ up to phase factors that are $p$th roots of unity. By multiplying $A_2$ with a suitable element in $H_1$ if necessary, we can ensure that $U_{A_2}$ commutes with $U_{A_1}, U_{B_1}$ and, consequently, the representations  of all elements in $H_1$. By the same reasoning as in the previous paragraph, there exists an element $B_2\in H$ such that
\begin{equation}
U_{A_2}U_{B_2} U_{A_2}^\dag U_{B_2}^\dag =\omega^{-1}.
\end{equation}
Note that $B_2$ cannot belong to $H_1$. By multiplying $B_2$ with a suitable element in $H_1$ if necessary, we may assume that $U_{B_2}$ commutes with $U_{A_1}$ and $U_{B_1}$.  Continuing this procedure, we can eventually find $n$ pairs of generators $A_1, B_1,\ldots, A_n,B_n$ of $H$ such that  $U_{A_1}, U_{B_1},\ldots, U_{A_n},U_{B_n}$ satisfy the canonical commutation relations
\begin{equation}\label{eq:CCRn2}
\begin{aligned}
U_{A_j}U_{B_k}U_{A_j}^{\dag} U_{B_k}^{\dag}&=\omega^{-\delta_{jk}},\\
\quad U_{A_j}U_{A_k}U_{A_j}^{\dag} U_{A_k}^{\dag}&=1,\\
U_{B_j}U_{B_k}U_{B_j}^{\dag} U_{B_k}^{\dag}&=1, \quad j,k=1,2,\ldots,n.
\end{aligned}
\end{equation}
According to the multipartite analogy of the Weyl's theorem \cite{BoltRW61I,BoltRW61II}, the representation must have degree $p^n$ and the image  must be (projectively) unitarily equivalent to the multipartite HW group in dimension $p^n$.
\end{proof}

\section{Proof of \lref{lem:3covariantQuasi}}
\begin{proof}
Suppose on the contrary that the SIC $\{\Pi_j\}$ is quasi-triply covariant.
Then it   is necessarily doubly covariant, that is, super-symmetric, given that the symmetry group of the SIC is a normal subgroup  of the extended symmetry group of index at most 2.
In addition, the normalized triple products $\tilde{T}_{jkl}$ for distinct $j,k,l$ can  take on at most two different values, which are  conjugate phase factors, say  $t$ and $t^*$. If $t=\pm1$, then the same proof of \lref{lem:3covariant} applies. Similar reasoning also shows that $t$ cannot equal $\pm\rmi$ when $d\geq3$.

Now suppose that $t^4\neq 1$.  Since the SIC $\{\Pi_j\}$ is  super-symmetric, the multiset $\{\tilde{T}_{mjk}\}_{m=1}^{d^2}$ is  invariant under complex conjugation and is independent of $j,k$ as long as  $j,k$ are distinct. It follows that the multiset contains two copies of 1 and $(d^2-2)/2$ copies of $t$ and $t^*$. This observation implies the lemma immediately when $d$ is odd.
 In general, choose a triple $j,k,l$ such that  $T_{jkl}=t$. According to \eref{eq:TripleProdP2} and the assumption  $t^4\neq 1$, if $m$ is distinct from $j,k,l$, then two of the three numbers $\tilde{T}_{mjk}, \tilde{T}_{mkl}, \tilde{T}_{mlj}$ are equal to $t$ and one equal to $t^*$. It follows that  the three multisets $\{\tilde{T}_{mjk}\}_{m=1}^{d^2}$, $\{\tilde{T}_{mkl}\}_{m=1}^{d^2}$, and $\{\tilde{T}_{mlj}\}_{m=1}^{d^2}$
contain at least $2(d^2-3)$ copies of $t$ in total. Therefore,
$2(d^2-3)\leq 3(d^2-2)/2$, that is, $d^2\leq 6$, which can never  hold when $d\geq3$.
\end{proof}

\section{\label{sec:PermuGroup}Permutation groups}
For the convenience of the reader, in this appendix we introduce several basic concepts and results about permutation groups that are relevant to  the study in the main text; see \rscite{DixoM96book, Came99book,KurzS04book} for more details.

Given a  finite set $\Omega=\{\alpha,\beta,\ldots\}$ with $n$ elements, the group composed  of all permutations on $\Omega$ is called the symmetric group on $\Omega$ and denoted by $S_\Omega$.  The group $S_\Omega$ is isomorphic to the symmetric group of the set $\{1,2,\ldots, n\}$, which  is usually denoted by $S_n$.
A group action of $G$  on $\Omega$ is a map from the Cartesian product $G\times \Omega$ to $\Omega$ that satisfies $1\alpha=\alpha$ and $g(h\alpha)=(gh)\alpha$, where 1 denotes the identity of $G$ (and  also the trivial group with only one element) and $g\alpha$ denotes the image of the pair $(g,\alpha)$ under the map. An \emph{orbit}  is the set of images of a point $\alpha\in \Omega$  (elements of $\Omega$ are usually referred to as points) under the action of $G$,  that is, $\{g\alpha: g\in G\}$. All orbits of the action form a partition of $\Omega$. The \emph{stabilizer} $G_\alpha$ of a point $\alpha$ is the group composed of all elements $g$ that leave $\alpha$ invariant, that is, $g\alpha=\alpha$.
The \emph{kernel} of the action is the group composed of all elements $g$ that act trivially on $\Omega$, that is, $g\alpha=\alpha$ for all $\alpha\in \Omega$. By definition, the kernel is the intersection $\cap_{\alpha\in\Omega}  G_\alpha$ of all point stabilizers.
Alternatively, a group action of $G$  on $\Omega$ is a homomorphism from $G$ to $S_\Omega$, and the kernel of the action coincides with the kernel of the homomorphism. The action is \emph{faithful} if the kernel is trivial.
A \emph{permutation group}  on $\Omega$ is a group $G$ that acts faithfully  on $\Omega$, which  can be identified as a subgroup of $S_\Omega$.  The \emph{degree} of the permutation group is the cardinality of $\Omega$.

A  permutation group $G$  on $\Omega$ is \emph{transitive} if every  element in $\Omega$ can be mapped to every other one under the action of $G$. In that case, all point stabilizers are conjugate to each other, and
the order of $G$ is equal to the product of the order of each point stabilizer and the cardinality of $\Omega$, that is, $|G|=|G_\alpha| |\Omega|$.  A transitive group is \emph{regular} if  each point stabilizer is trivial, in which case $|G|=|\Omega|$.
The group $G$ is \emph{$k$-transitive} if every ordered $k$-tuple of distinct elements of $\Omega$ can be mapped to every other such $k$-tuple.
When $k>1$, a group  is $k$-transitive if and only if it is transitive and each point stabilizer is $(k-1)$-transitive on the remaining points.
A $k$-transitive group  is \emph{sharply $k$-transitive} if  each $k$-point stabilizer is trivial.
The group $G$ is \emph{$k$-homogeneous} if every unordered $k$-tuple of distinct elements  can be mapped to every other such $k$-tuple. By definition,  a permutation group of degree $n$ is $k$-homogeneous if and only if it is $(n-k)$-homogeneous. In addition,
every $k$-transitive permutation group is $k$-homogeneous.

A \emph{block} $\Delta$ (also called a set of imprimitivity)  of the action of $G$ on $\Omega$ is a nonempty subset of $\Omega$ such that $g\Delta:=\{g\alpha: \alpha\in \Delta\}$ for any $g\in G$ is either identical with $\Delta$ or disjoint from $\Delta$.  The block is nontrivial if it is a proper subset of $\Omega$ that contains more than one elements. A transitive permutation group  is \emph{imprimitive} if there exists a nontrivial block and \emph{primitive} otherwise. Alternatively, the group  is primitive if no nontrivial partition of $\Omega$ is left invariant, where a  partition is nontrivial if it has at least two components each of which has at least two points.
 Primitive permutation groups play a crucial role in the study of permutation groups. Their basic properties are listed below for easy reference.
\begin{enumerate}
\item A permutation group $G$ is primitive if and only if each point stabilizer is a maximal subgroup of $G$.

\item Every normal subgroup $N\neq 1$ of a primitive permutation group $G$ on $\Omega$ acts transitively on $\Omega$.

\item Every 2-transitive permutation group is primitive.
\end{enumerate}
\begin{remark}
A \emph{maximal subgroup} $M$ of a group $G$ is a proper subgroup that is not contained in any other proper subgroup.
\end{remark}

To better understand the properties of primitive and 2-transitive permutation groups, we need to introduce several new concepts. A \emph{minimal normal subgroup} $N$ of a group $G$ is a normal subgroup other than the identity that contains no other nontrivial normal subgroup of $G$. Any minimal normal subgroup  is a direct product of isomorphic simple groups. Any two distinct minimal normal subgroups of a given group have a trivial intersection and commute with each other. The \emph{socle}  of a group $G$ is the product of all minimal normal subgroups of $G$ and is denoted by $\soc(G)$; it is a characteristic subgroup of $G$. Recall that a characteristic subgroup of $G$ is a subgroup that is invariant under all automorphisms of $G$. The structure of minimal normal subgroups of a primitive permutation group is characterized by Theorem 4.3B in \rcite{DixoM96book}, as reproduced here.
\begin{theorem}\label{thm:primitive}
If $G$ is a finite primitive permutation group on $\Omega$, and $N$ a minimal normal subgroup of $G$, then exactly one of the following holds:
\begin{enumerate}
\item $N$ is a regular elementary abelian group, and $\soc(G)=N=C_G(N)$.

\item $N$ is a regular nonabelian group, $C_G(N)$ is a minimal normal subgroup of $G$ which is permutation isomorphic to $N$, and $\soc(G)=N\times C_G(N)$.

\item $N$ is nonabelian, $C_G(N)=1$, and $\soc(G)=N$.
\end{enumerate}
\end{theorem}
\begin{remark}
An elementary abelian group is the direct product of cyclic groups of the same  prime order.  Here $C_G(N)$ denotes the centralizer of $N$ in $G$. Two permutation groups on $\Omega$ are \emph{permutation isomorphic} if they are conjugate to each other under $S_\Omega$. This theorem implies that a primitive permutation group has at most two minimal normal subgroups. It has only one minimal normal subgroup if the socle is abelian. Note that the socle of a primitive permutation group is abelian if and only if one  minimal normal subgroup is abelian.
\end{remark}

\section{Proof of \lref{lem:index2Normalsub}}
\begin{proof}
Note that any subgroup of index 2 is normal,  that any 2-transitive permutation group is primitive, and that any nontrivial normal subgroup of a primitive  permutation group is transitive \cite{DixoM96book,Came99book,KurzS04book}.
It follows  that $H$ is a transitive normal subgroup of $G$. In addition, the point stabilizer  $H_\alpha$ of any point $\alpha\in \Omega$ has either one orbit or two orbits of equal length on the remaining points of $\Omega$.
In the former case, $H$ is 2-transitive and is thus primitive; in the latter case, it is also straightforward to show that $H$ is primitive. According to \thref{thm:primitive}, $H$ has  either one or  two minimal normal subgroups, and in the latter case the two minimal normal subgroups are nonabelian regular and are centralizers of each other\footnote{The author is grateful to Dragomir \v{Z} {\DJ}okovi\'{c} for simplifying the proof of \lref{lem:index2Normalsub}.}.

Let $N$ be the unique  minimal normal subgroup of $G$. Then   $N$ is contained in $H$ and contains one of the  minimal normal subgroups of $H$, say $M$. According to Burnside's theorem (\thref{thm:Burnside} in the main text) on 2-transitive permutation groups \cite{DixoM96book,Came99book},  $N$ is either elementary abelian regular or nonabelian simple, and the centralizer $C_G(N)$ is either identical with $N$ or trivial accordingly. Therefore, $M$ must be identical with $N$ since  it is either a transitive subgroup of the regular group $N$ or a normal subgroup of the simple group $N$.  Consequently, the centralizer $C_H(M)$ is either identical with $M$ or trivial.
It follows from \thref{thm:primitive} that $H$ has a unique minimal normal subgroup, which coincides with the unique minimal normal subgroup of $G$.
\end{proof}

\section{Proof of \lref{lem:2transitiveAS}}
All 2-transitive permutation groups have been classified by Huppert \cite{Hupp57} and Hering \cite{Heri74,Heri85} (see also \rscite{DixoM96book,Came99book}), with the aid of the classification of finite simple groups (CFSG) \cite{ConwCNP85,Wils09book}.
Almost simple 2-transitive permutation groups are listed in  Table 7.4 in  \rcite{Came99book}. According to this table, if the minimal normal subgroup (socle) $N$ of the group is not  the alternating group $A_n$ for $n\geq 5$, then the degree $n$ can only take on the following values:
\begin{enumerate}
\item $(q^k-1)/(q-1)$ with $k\geq 2$ and $(k,q)\neq(2,2), (2,3)$;
\item $2^{2k-1}+2^{k-1}$ with $k\geq 3$;
\item $2^{2k-1}-2^{k-1}$ with $k\geq 3$;
\item $q^3+1$ with  $q\geq3$;
\item $q^2+1$ with $q=2^{2k+1}>2$;
\item 11, 12, 15, 22, 23, 24, 28, 176, 276;
\end{enumerate}
where $q$ is a prime power and $k$ a positive integer.

In case 1, $N$ is isomorphic to $\PSL{k}{q}$. According to \lref{lem:NageL} below, the degree $n$ is a perfect square if and only if  $(k,q)=(2,8)$, $(4,7)$, or $(5,3)$, in which case  $n$ is equal to $3^2$, $20^2$, or $11^2$ accordingly.

In case 2, $N$ is isomorphic to $\Sp{2k}{2}$. According to \lref{lem:perfectS3} below, the degree $n$ is a perfect square if and only if $k=3$, in which case $n$ is equal to~$6^2$.

In case 3, the degree can never be a perfect square according to \lref{lem:perfectS3}. In cases 4 and 5, the degree can never be a perfect square according to \lref{lem:perfectS1} below. In case 6, the degree can never be a perfect square by direct inspection.  This observation completes the proof of \lref{lem:2transitiveAS}.

\begin{lemma}\label{lem:perfectS1}
Suppose $q$ is a prime power and $m$ a positive integer.  Then $q+1$ is a perfect square if and only if $q=3$ or $q=8$;  $q^m+1$ is a perfect square if and only if $(m,q)=(1,3)$, $(1,8)$, or $(3,2)$. \end{lemma}
\begin{proof}
Obviously  $q+1$ is a perfect square when $q=3$ or $q=8$.   Suppose $q=p^j$ and $q+1=b^2$, where $p$ is a prime,  and $j, b$ are positive integers.  Then $p^j=(b-1)(b+1)$. If $q>3$ then $b>2$, so $p$ divides both $b-1$ and $b+1$ and thus  equals 2. Consequently, both $b-1$ and $b+1$ must be powers of 2. It follows that $b=3$ and $q=8$. The second part of the lemma is an immediate consequence of the first part given that  $q^m$ is also a prime power.
\end{proof}

\begin{lemma}\label{lem:perfectS2}
Suppose $q$ is a power of the prime $p$ and $j\geq k$ are nonnegative integers.  Then $q^j+q^k$ is a perfect square if and only if one of the following conditions is satisfied
\begin{enumerate}
\item $j=k$ is odd and $q$ is an odd power of 2.
\item $q^k$ is a perfect square and $(j-k,q)=(1,3)$, $(1,8)$, or $(3,2)$.
\end{enumerate}
\end{lemma}

\begin{proof}
If $j=k$, then $q^j+q^k=2q^j$, which is a perfect square if and only if $j$ is odd and $q$ is an odd power of 2. If $j>k$, then $q^j+q^k=q^k(q^m+1)$ with $m=j-k$. If $q^k$ is not a perfect square, then $q^j+q^k$ is a perfect square if and only if $p(q^m+1)$ is a perfect square, which is impossible. Otherwise, $q^m+1$ is a perfect square, so  $(m,q)=(1,3)$, $(1,8)$, or $(3,2)$ according to \lref{lem:perfectS1}.
\end{proof}

\begin{lemma}\label{lem:perfectS3}
Suppose $j>k$ are nonnegative  integers. Then $2^j+2^k$  is a perfect square if and only if $k$ is even and $j=k+3$;  $2^j-2^k$  is a perfect square if and only if $k$ is even and $j=k+1$.
\end{lemma}
\begin{proof} The first statement follows from \lref{lem:perfectS2}.
If $k$ is odd, then $2^j-2^k$ is a perfect square if and only if $2(2^m-1)$ with $m=j-k$ is a perfect square, which is impossible. Otherwise, $2^j-2^k$ is a perfect square if and only if $(2^m-1)$  is a perfect square, say  $2^m-1=b^2$ with $b$ a positive integer.  This is possible if and only if $m=1$ given that $b^2+1$ is not divisible by 4.
\end{proof}

\begin{lemma}\label{lem:NageL}
Suppose $q$ is a prime power and $k\geq2$ a positive integer.  Then $(q^k-1)/(q-1)$ is a perfect square if and only if the pair $(k,q)$ takes on one of the four possible values  $(2,3)$, $(5,3)$, $(4,7)$, and $(2,8)$.
\end{lemma}
\begin{proof}
The equation
\begin{equation}
\frac{q^k-1}{q-1}=d^2
\end{equation}
with $d$ a positive integer is a special instance of  the Nagell-Ljunggren equation. If $k\geq 3$, then the only integer solutions are $(k,q,d)=(5,3,11)$ and $(k,q,d)=(4,7,20)$ \cite{Ljun43, Ribe94,BugeM07}\footnote{The author is grateful to Gergely Harcos  for helping proving \lref{lem:NageL} by introducing the concept of Nagell-Ljunggren equation and  \rcite{Ljun43} and to Daniel El-Baz for recommending \rcite{Ribe94} in response to a question posed by the author on MathOverflow.}.

If   $k=2$, then $(q^k-1)/(q-1)=q+1$ is a perfect square  if and only if $q=3$ or $q=8$ according to \lref{lem:perfectS1}.
\end{proof}

\bibliographystyle{elsarticle-num}
\bibliography{all_references}

\end{document}